\documentclass{article}

\usepackage{microtype}
\usepackage{graphicx}
\usepackage{subcaption}
\usepackage{booktabs}
\usepackage{tabularx}
\usepackage{hyperref}

\usepackage[preprint]{icml2026}

\usepackage{amsmath}
\usepackage{amssymb}
\usepackage{mathtools}
\usepackage{amsthm}
\usepackage{bm}
\usepackage{bbm}

\usepackage{algorithm}
\usepackage{algorithmic}

\usepackage{tikz}
\usetikzlibrary{positioning,arrows.meta,shapes}

\usepackage[capitalize,noabbrev]{cleveref}

\crefname{assumption}{Assumption}{Assumptions}
\Crefname{assumption}{Assumption}{Assumptions}
\crefname{example}{Example}{Examples}
\Crefname{example}{Example}{Examples}
\crefname{remark}{Remark}{Remarks}
\Crefname{remark}{Remark}{Remarks}

\theoremstyle{plain}
\newtheorem{theorem}{Theorem}[section]
\newtheorem{proposition}[theorem]{Proposition}
\newtheorem{lemma}[theorem]{Lemma}
\newtheorem{corollary}[theorem]{Corollary}
\theoremstyle{definition}

\newtheorem{assumption}[theorem]{Assumption}
\newtheorem{example}[theorem]{Example}
\theoremstyle{remark}
\newtheorem{remark}[theorem]{Remark}

\newcommand{\E}{\mathbb{E}}
\newcommand{\Pp}{\mathbb{P}}
\newcommand{\R}{\mathbb{R}}
\newcommand{\Var}{\operatorname{Var}}
\newcommand{\Cov}{\operatorname{Cov}}
\newcommand{\indep}{\perp\!\!\!\perp}

\newcommand{\cX}{\mathcal{X}}

\newcommand{\cD}{\mathcal{D}}

\icmltitlerunning{OVB Sensitivity Analysis for Trial Generalization}

\begin{document}

\twocolumn[
\icmltitle{Omitted-Variable Sensitivity Analysis for \\
Generalizing Randomized Trials}

% Author block
\begin{icmlauthorlist}
\icmlauthor{Amir Asiaee}{vumc}
\icmlauthor{Samhita Pal}{vumc}
\icmlauthor{Jared D.\ Huling}{umn}
\end{icmlauthorlist}

\icmlaffiliation{vumc}{Department of Biostatistics, Vanderbilt University Medical Center, Nashville, TN, USA}
\icmlaffiliation{umn}{Division of Biostatistics and Health Data Science, University of Minnesota, Minneapolis, MN, USA}
\icmlcorrespondingauthor{Amir Asiaee}{amir.asiaeetaheri@vumc.org}

\icmlkeywords{Causal Inference, External Validity, Sensitivity Analysis, Transportability, Machine Learning}

\vskip 0.3in
]

\printAffiliationsAndNotice{}

%%%%%%%%%%%%%%%%%%%%%%%%%%%%%%%%
% ABSTRACT
%%%%%%%%%%%%%%%%%%%%%%%%%%%%%%%%
\begin{abstract}
Randomized controlled trials (RCTs) yield internally valid causal effect estimates, but generalizing these results to target populations with different characteristics requires an untestable \emph{selection ignorability} assumption: conditional on observed covariates, trial participation must be independent of potential outcomes.
This assumption fails when unobserved effect modifiers are distributed differently between trial and target populations.

We develop a sensitivity analysis framework for trial generalization grounded in \emph{omitted variable bias} (OVB).
Our key theoretical contribution is an exact decomposition showing that external-validity bias equals \textbf{moderation strength} $\times$ \textbf{moderator imbalance}: (i) how strongly an unobserved variable shifts the treatment effect, times (ii) how differently that variable is distributed across populations after covariate adjustment.
We introduce scale-free sensitivity parameters based on partial $R^2$ values, enabling closed-form bounds and benchmarking against observed covariates---practitioners can assess whether conclusions would change if an unobserved moderator were ``as strong as'' a particular observed variable.
Simulations demonstrate that our bounds achieve nominal coverage and remain conservative under model misspecification, while comparisons with alternative sensitivity frameworks highlight the interpretive advantages of the OVB decomposition.
\end{abstract}

%%%%%%%%%%%%%%%%%%%%%%%%%%%%%%%%
% 1. INTRODUCTION
%%%%%%%%%%%%%%%%%%%%%%%%%%%%%%%%
\section{Introduction}
\label{sec:intro}

Machine learning models trained on data from one distribution often fail when deployed to populations with different characteristics---a phenomenon known as \emph{distribution shift} \citep{quinonero2009dataset,koh2021wilds}.
In causal inference, an analogous challenge arises when we attempt to \emph{generalize} or \emph{transport} treatment effect estimates from a randomized controlled trial (RCT) to a target population that differs from the experimental sample.
Even when an RCT provides internally valid estimates of causal effects, these estimates may not apply to the population where a policy or intervention will actually be deployed \citep{stuart2011use,degtiar2023review}.
This problem---the gap between internal and external validity---is increasingly recognized as a fundamental barrier to evidence-based decision-making \citep{westreich2017transportability,colnet2024causal}.

\paragraph{The generalization problem.}
Consider a pharmaceutical company that conducts a clinical trial at urban academic medical centers, but plans to market the resulting drug nationally.
Or a tech company that A/B tests a new feature on power users who opt into beta programs, but will deploy it to all users.
In both cases, trial participants may systematically differ from the target population in ways that affect how they respond to treatment.
When treatment effects are \emph{heterogeneous}---varying across individuals based on their characteristics---these differences can lead to substantial discrepancies between the trial sample average treatment effect (SATE) and the target average treatment effect (TATE).

\paragraph{Standard solutions and their limitations.}
The methodological literature on \emph{generalizability} and \emph{transportability} provides estimators for the TATE by combining trial outcomes with covariate information from a target sample \citep{cole2010generalizing,stuart2011use,buchanan2018generalizing,dahabreh2019extending,dahabreh2021study}.
These methods adjust for observed differences between trial and target populations using techniques such as:
\begin{itemize}
    \item \textbf{Inverse probability weighting (IPW)}: Reweight trial observations by the inverse odds of trial participation.
    \item \textbf{Outcome modeling (g-formula)}: Fit a model predicting outcomes from treatment and covariates in the trial, then average predictions over the target covariate distribution.
    \item \textbf{Doubly robust methods}: Combine weighting and outcome modeling for robustness to model misspecification.
\end{itemize}

All such estimators require an \emph{untestable} identification assumption: after conditioning on observed covariates $X$, trial participation $S$ must be independent of potential outcomes $Y(a)$.
This \emph{selection ignorability} assumption fails whenever there exist unmeasured \emph{effect modifiers}---variables that both (i) influence how individuals respond to treatment and (ii) are distributed differently between trial and target populations.

\paragraph{Why sensitivity analysis?}
In practice, researchers can never be certain that all relevant effect modifiers have been measured.
Subject-matter knowledge may suggest candidate unmeasured moderators, but their precise distributions and effect-modification strengths remain unknown.
This motivates \emph{sensitivity analysis}: systematic exploration of how conclusions change under hypothetical violations of the identification assumptions.

A well-designed sensitivity analysis should answer questions such as:
\begin{itemize}
    \item How strong would an unmeasured effect modifier need to be to change the sign of the transported effect?
    \item How does the plausible range of the TATE expand as we allow for progressively larger violations?
    \item Are the violations required to overturn our conclusions plausible given domain knowledge?
\end{itemize}

\paragraph{Our contribution.}
We develop a sensitivity analysis for trial generalization based on \emph{omitted variable bias} (OVB).
Our approach yields transparent, interpretable bounds that decompose external-validity bias into two distinct components:
\textbf{bias = moderation strength $\times$ moderator imbalance}.
Specifically, bias arises from (i) \emph{how strongly} an unobserved variable $U$ modifies the treatment effect, and (ii) \emph{how differently} $U$ is distributed between trial and target populations after adjusting for observed covariates.

We make the following contributions:
\begin{enumerate}
    \item \textbf{OVB decomposition for external validity} (\cref{sec:ovb}): Under a linear effect-moderation model, we derive an exact identity expressing the TATE estimation error as a product of moderator strength and moderator imbalance.
    \item \textbf{Partial $R^2$ sensitivity parameterization} (\cref{sec:r2}): We introduce scale-free sensitivity parameters based on partial $R^2$ values, enabling comparisons across outcomes and benchmarking against observed covariates.
    \item \textbf{Robustness summaries and benchmarking} (\cref{sec:robustness}): We define ``robustness values'' quantifying the minimum confounding strength needed to change substantive conclusions, and show how to benchmark these against observed effect modifiers.
    \item \textbf{Practical workflow and experiments} (\cref{sec:estimation,sec:experiments}): We provide a complete estimation workflow and demonstrate the method's performance on synthetic and semi-synthetic data.
\end{enumerate}
All proofs are deferred to \cref{app:proofs}.

%%%%%%%%%%%%%%%%%%%%%%%%%%%%%%%%
% 2. RELATED WORK
%%%%%%%%%%%%%%%%%%%%%%%%%%%%%%%%
\section{Related Work}
\label{sec:related}

\paragraph{Generalizability and transportability.}
The statistical literature on extending causal inferences from trials to target populations has grown substantially.
\citet{cole2010generalizing} introduced inverse probability of sampling weights for generalization, building on the potential outcomes framework.
\citet{stuart2011use} developed propensity-score-based methods and provided practical guidance, while \citet{tipton2013improving} proposed subclassification approaches for educational experiments.
\citet{hartman2015sate} showed how to combine experimental and observational data to estimate population treatment effects, and \citet{kern2016assessing} systematically compared methods including BART and weighting.
\citet{buchanan2018generalizing} extended these methods to complex survey designs, and \citet{dahabreh2019extending,dahabreh2021study} provided a comprehensive framework covering identification, estimation, and study design considerations.
\citet{lesko2017generalizing} offered a clear potential-outcomes perspective on the key assumptions.
From a causal graphical perspective, \citet{pearl2011transportability} and \citet{bareinboim2016causal} formalized when and how causal effects can be transported across environments using selection diagrams.
\citet{rudolph2017robust} developed TMLE-based robust methods for transporting effects across sites.
Recent reviews synthesize this literature: \citet{degtiar2023review} provide a statistical overview, \citet{ling2023overview} focus on practical applications, and \citet{colnet2024causal} discuss connections to machine learning and policy learning \citep{athey2021policy}.

\paragraph{Sensitivity analysis for external validity.}
While sensitivity analysis is well-developed for observational studies \citep{rosenbaum2002observational,robins2000sensitivity,ding2016sensitivity}, methods specifically addressing generalization are more recent.
\citet{nguyen2018sensitivity} consider the case where a moderator is observed in the trial but not in the target dataset, developing bounds based on the moderator's effect-modification strength.
\citet{nie2021covariate} propose optimization-based bounds combining marginal sensitivity models with covariate balancing constraints, allowing for flexible but computationally intensive analysis.
\citet{dahabreh2023sensitivity} parameterize violations via bias functions on the target counterfactual means, providing a general framework but with parameters that lack direct interpretation.
Most closely related to our work, \citet{huang2024sensitivity} develop a two-parameter sensitivity analysis for weighted generalization estimators with benchmarking capabilities; our approach differs by grounding the analysis in an explicit OVB decomposition that separates moderation strength from moderator imbalance.

Our contribution connects the generalization literature to the partial-$R^2$ sensitivity framework of \citet{cinelli2020making}, yielding closed-form bounds with transparent ``strength $\times$ imbalance'' interpretation and enabling direct benchmarking against observed covariates.

\begin{table*}[t]
    \centering
    \small
    \caption{Positioning of our approach relative to representative sensitivity analyses for trial generalization. The goal is not to replace prior frameworks, but to provide an OVB-based lens that yields closed-form, benchmarkable sensitivity summaries.}
    \label{tab:positioning}
    \begin{tabular}{lcccp{5.5cm}}
        \toprule
        Method & Sensitivity parameterization & Closed form & Benchmarking & Primary focus \\
        \midrule
        \citet{nguyen2018sensitivity} & Moderation $\times$ imbalance & Yes & Limited & Missing moderators observed in trial only \\
        \citet{nie2021covariate} & Marginal sensitivity model & No & No & Outcome shift via odds-ratio constraints \\
        \citet{dahabreh2023sensitivity} & Bias functions & Yes & No & Selection bias in transport estimators \\
        \citet{huang2024sensitivity} & Two-parameter model & Yes & Yes & Sensitivity of weighted estimators \\
        \textbf{Ours} & OVB + partial-$R^2$ & Yes & Yes & External validity via OVB decomposition \\
        \bottomrule
    \end{tabular}
\end{table*}

\paragraph{OVB and partial-$R^2$ sensitivity.}
In observational studies, omitted variable bias provides a foundational framework for understanding confounding \citep{cochran1973controlling,rosenbaum1983assessing}.
\citet{cinelli2020making} transformed this framework by introducing partial $R^2$ parameters, which are scale-free and enable transparent benchmarking; their \texttt{sensemakr} package has been widely adopted for sensitivity analysis in regression settings.
\citet{blackwell2014selection} and \citet{oster2019unobservable} provide related approaches in political science and economics, respectively.
\citet{chernozhukov2024ovb} further develop a general OVB theory for a broad class of causal machine learning targets (including covariate-shift policy effects) using Riesz representers and debiased ML inference.
Our trial generalization setting can be seen as an external-validity analogue: latent variables can simultaneously drive trial participation and induce treatment-effect heterogeneity.
\cref{app:riesz} sketches a mapping between our bounds and the general Riesz-representer OVB bounds.

%%%%%%%%%%%%%%%%%%%%%%%%%%%%%%%%
% 3. SETUP AND BASELINE IDENTIFICATION
%%%%%%%%%%%%%%%%%%%%%%%%%%%%%%%%
\section{Setup and Baseline Identification}
\label{sec:setup}

\subsection{Data Structure and Notation}

Let $S \in \{0, 1\}$ indicate population membership, where $S = 1$ denotes the randomized trial and $S = 0$ denotes the target population.
Each unit $i$ has:
\begin{itemize}
    \item Baseline covariates $X_i \in \cX \subseteq \R^p$
    \item Binary treatment assignment $A_i \in \{0, 1\}$
    \item Observed outcome $Y_i \in \R$
    \item Potential outcomes $Y_i(0)$ and $Y_i(1)$ under control and treatment
\end{itemize}

We observe two datasets:
\begin{align}
    \cD^{\text{trial}} &= \{(X_i, A_i, Y_i) : S_i = 1\}_{i=1}^{n_r} \\
    \cD^{\text{target}} &= \{X_j : S_j = 0\}_{j=1}^{n_o}
\end{align}
In the trial, treatment is randomized; in the target, we observe only covariates (no treatments or outcomes).
This ``non-nested'' design is common in practice, though our methods extend to nested designs where trial participants are sampled from the target \citep{dahabreh2021study}.

Figure~\ref{fig:dag} illustrates the causal structure.
The key challenge is that an unobserved variable $U$ may simultaneously (i) modify treatment effects (the $U \to Y$ path that varies with $A$) and (ii) have a different distribution between trial and target populations (the dashed $S \dashleftarrow U$ relationship, indicating that $P(U \mid X, S)$ depends on $S$).

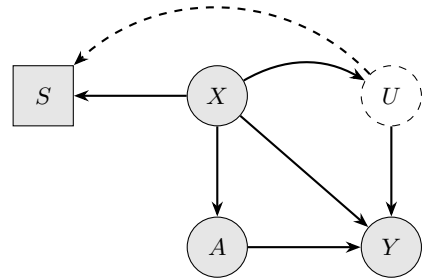
\begin{figure}[t]
\centering
\begin{tikzpicture}[
    node distance=1.2cm and 1.5cm,
    every node/.style={circle, draw, minimum size=0.8cm, font=\small},
    observed/.style={fill=gray!20},
    latent/.style={fill=white, dashed},
    selection/.style={rectangle, fill=gray!20},
    arr/.style={-{Stealth[length=2mm]}, thick}
]
% Nodes
\node[observed] (X) {$X$};
\node[latent, right=of X] (U) {$U$};
\node[observed, below=of X] (A) {$A$};
\node[observed, below=of U] (Y) {$Y$};
\node[selection, left=of X] (S) {$S$};

% Edges
\draw[arr] (X) -- (Y);
\draw[arr] (X) -- (A);
	\draw[arr] (A) -- (Y);
	\draw[arr] (U) -- (Y);
	\draw[arr] (X) -- (S);
	\draw[arr, dashed] (U) to[bend right=45] (S);
	\draw[arr] (X) to[bend left=30] (U);
\end{tikzpicture}
\caption{Causal structure for trial generalization with an unobserved effect modifier $U$.
Observed variables ($X$, $A$, $Y$, $S$) are shaded; latent $U$ is unshaded with a dashed border.
The dashed edge $U \dashrightarrow S$ indicates that $U$'s distribution differs between trial ($S=1$) and target ($S=0$) populations, i.e., $P(U \mid X, S=1) \neq P(U \mid X, S=0)$.
If $U$ also modifies treatment effects ($U \to Y$ interaction with $A$), selection ignorability fails and standard transport estimators are biased.}
\label{fig:dag}
\end{figure}

\subsection{Target Estimand}

Our goal is to estimate the \textbf{target average treatment effect} (TATE):
\begin{equation}
\label{eq:tate}
    \tau^* := \E[Y(1) - Y(0) \mid S = 0]
\end{equation}
This is the average causal effect of treatment in the population where the intervention will be deployed, not the trial population.

For comparison, the \textbf{sample average treatment effect} (SATE) in the trial is:
\begin{equation}
    \tau^{\text{SATE}} := \E[Y(1) - Y(0) \mid S = 1]
\end{equation}
When treatment effects are heterogeneous and the trial and target populations differ in their covariate distributions, $\tau^*$ and $\tau^{\text{SATE}}$ may differ substantially.

\subsection{Trial Internal Validity}

We maintain standard assumptions ensuring the trial provides valid causal estimates for its own population.

\begin{assumption}[Trial internal validity]
\label{assump:internal}
Within the trial ($S = 1$):
\begin{enumerate}
    \item \textbf{Consistency}: $Y = Y(A)$ almost surely.
    \item \textbf{Randomization}: $(Y(0), Y(1)) \indep A \mid X, S = 1$.
    \item \textbf{Positivity}: $0 < \Pp(A = 1 \mid X, S = 1) < 1$ almost surely.
\end{enumerate}
\end{assumption}

Under \cref{assump:internal}, the conditional average treatment effect (CATE) in the trial population is identified:
\begin{equation}
\begin{split}
    \tau^r(x) &:= \E[Y(1) - Y(0) \mid X = x, S = 1] \\
             &= \mu_1^r(x) - \mu_0^r(x)
\end{split}
\end{equation}
where $\mu_a^r(x) := \E[Y \mid A = a, X = x, S = 1]$ is the conditional mean outcome in the trial.

\subsection{Baseline Transport Assumption}

To identify $\tau^*$ from the available data, standard approaches assume:

\begin{assumption}[Selection ignorability]
\label{assump:sel-ign}
For each $a \in \{0, 1\}$:
\begin{equation}
    Y(a) \indep S \mid X
\end{equation}
\end{assumption}

\cref{assump:sel-ign} states that, conditional on observed covariates, potential outcomes are identically distributed across trial and target populations.
Equivalently, $X$ captures all variables that both (i) affect the outcome and (ii) are distributed differently between populations.

\begin{proposition}[Identification under selection ignorability]
\label{prop:identification}
Under \cref{assump:internal,assump:sel-ign} and the additional positivity assumption $\Pp(S = 1 \mid X) > 0$ for all $x$ in the target support, the TATE is identified by:
\begin{equation}
\label{eq:transport}
    \tau^* = \E[\tau^r(X) \mid S = 0] = \E[\mu_1^r(X) - \mu_0^r(X) \mid S = 0]
\end{equation}
\end{proposition}

The proof follows directly from the tower property and \cref{assump:sel-ign}.
In practice, we estimate $\mu_a^r(x)$ from trial data and average over the target covariate distribution.

\begin{remark}[Alternative estimators]
Equation~\eqref{eq:transport} suggests an outcome-modeling (g-formula) estimator.
Alternatively, one can use inverse probability weighting with selection odds:
\begin{equation}
    \widehat{\tau}^{\text{IPW}} = \frac{\sum_{i: S_i = 1} w(X_i) A_i Y_i}{\sum_{i: S_i = 1} w(X_i) A_i} - \frac{\sum_{i: S_i = 1} w(X_i) (1-A_i) Y_i}{\sum_{i: S_i = 1} w(X_i) (1-A_i)}
\end{equation}
where $w(x) = \frac{1-\widehat{e}(x)}{\widehat{e}(x)}$ are the inverse odds of trial participation weights and $\widehat{e}(x) = \widehat{\Pp}(S = 1 \mid X = x)$ \citep{buchanan2018generalizing,westreich2017transportability}.
Doubly robust estimators combine outcome modeling and weighting for robustness to misspecification of either model \citep{dahabreh2019extending}.
\end{remark}

%%%%%%%%%%%%%%%%%%%%%%%%%%%%%%%%
% 4. OVB SENSITIVITY MODEL
%%%%%%%%%%%%%%%%%%%%%%%%%%%%%%%%
\section{An OVB Sensitivity Model for Generalization}
\label{sec:ovb}

We now relax \cref{assump:sel-ign} by allowing for an unobserved effect modifier $U$ whose distribution differs between trial and target populations.

\subsection{Latent Moderator Model}

\begin{assumption}[Latent moderator bridge]
\label{assump:bridge}
There exists an unobserved random variable $U$ such that for each $a \in \{0, 1\}$:
\begin{equation}
    Y(a) \indep S \mid (X, U)
\end{equation}
\end{assumption}

\cref{assump:bridge} weakens \cref{assump:sel-ign}: selection ignorability holds conditional on $(X, U)$ rather than $X$ alone.
If $U$ were observed, we could adjust for it and identify $\tau^*$.
The problem is that $U$ is unmeasured, so we cannot directly control for it.

\begin{example}[Unmeasured effect modifiers]
In a clinical trial for a cardiovascular drug:
\begin{itemize}
    \item $U$ = genetic variants affecting drug metabolism
    \item $U$ = patient adherence patterns
    \item $U$ = access to complementary care
\end{itemize}
These may modify treatment effects and be distributed differently across trial sites vs.\ the national population.
\end{example}

\subsection{Linear Effect-Moderation Model}

To derive tractable sensitivity bounds, we impose a linear structure on how $U$ affects potential outcomes.

\begin{assumption}[Linear effect modification]
\label{assump:linear}
For $a \in \{0, 1\}$, the conditional mean potential outcome satisfies:
\begin{equation}
\label{eq:linear-y}
    \E[Y(a) \mid X, U, S] = m_a(X) + \eta_a(X) \cdot U
\end{equation}
where $m_a(X)$ captures the $X$-dependent baseline and $\eta_a(X)$ captures effect modification by $U$.
Without loss of generality, we center $U$ so that $\E[U \mid X] = 0$.
\end{assumption}

The key quantity is the \textbf{moderation strength}:
\begin{equation}
    \beta(X) := \eta_1(X) - \eta_0(X)
\end{equation}
This measures how a one-unit increase in $U$ changes the treatment effect at covariate value $X$.

\begin{remark}[Interpretation]
\cref{assump:linear} is a first-order Taylor approximation to any smooth conditional mean function.
It parallels the linear sensitivity models used in observational studies \citep{cinelli2020making} and captures the key ``strength $\times$ imbalance'' structure.
We discuss nonlinear extensions in \cref{sec:discussion}.
\end{remark}

\subsection{The OVB Decomposition}

Define the \textbf{moderator imbalance} between trial and target:
\begin{equation}
\label{eq:delta-u}
    \Delta_U(X) := \E[U \mid X, S = 0] - \E[U \mid X, S = 1]
\end{equation}
This measures how differently $U$ is distributed across populations, conditional on $X$.
If $\Delta_U(X) = 0$ for all $X$, then \cref{assump:sel-ign} holds despite $U$ being unmeasured.

Define the \textbf{$X$-adjusted transport estimand}:
\begin{equation}
\label{eq:tau-x}
    \tau_X := \E[\tau^r(X) \mid S = 0]
\end{equation}
This is what standard transport estimators target.

\begin{lemma}[External-validity OVB identity]
\label{lem:ovb-identity}
Under \cref{assump:internal,assump:bridge,assump:linear}:
\begin{equation}
\label{eq:tate-ovb}
    \tau^* = \tau_X + \E[\beta(X) \cdot \Delta_U(X) \mid S = 0]
\end{equation}
In particular, if $\beta(X) \equiv \beta$ is constant:
\begin{equation}
\label{eq:tate-ovb-const}
    \tau^* = \tau_X + \beta \cdot \Delta_U^*
\end{equation}
where $\Delta_U^* := \E[\Delta_U(X) \mid S = 0]$.
\end{lemma}

\noindent\emph{Proof in \cref{app:proof-ovb-identity}.}

\paragraph{Interpretation.}
\cref{lem:ovb-identity} decomposes external-validity bias into two conceptually distinct components:
\begin{enumerate}
    \item \textbf{Moderation strength} $\beta(X)$: How strongly does $U$ modify treatment effects?
    \item \textbf{Moderator imbalance} $\Delta_U(X)$: How differently is $U$ distributed between trial and target (after $X$-adjustment)?
\end{enumerate}
If either component is zero, selection ignorability holds and $\tau^* = \tau_X$.
Bias requires \emph{both} treatment-effect heterogeneity driven by $U$ \emph{and} differential distribution of $U$ across populations.

\subsection{Simple Sensitivity Interval}

\cref{lem:ovb-identity} immediately yields sensitivity bounds.

\begin{corollary}[Raw sensitivity interval]
\label{cor:simple-interval}
Suppose $|\beta(X)| \leq \Gamma$ and $|\Delta_U(X)| \leq \Lambda$ almost surely.
Then:
\begin{equation}
    \tau^* \in [\tau_X - \Gamma \Lambda, \; \tau_X + \Gamma \Lambda]
\end{equation}
\end{corollary}

The parameters $(\Gamma, \Lambda)$ have direct interpretations:
\begin{itemize}
    \item $\Gamma$: Maximum effect modification---how much can a one-unit increase in $U$ change the treatment effect?
    \item $\Lambda$: Maximum imbalance---how much can the mean of $U$ differ between trial and target at any $X$?
\end{itemize}

While intuitive, these parameters are scale-dependent and hard to benchmark.
We address this next.

%%%%%%%%%%%%%%%%%%%%%%%%%%%%%%%%
% 5. PARTIAL R^2 PARAMETERIZATION
%%%%%%%%%%%%%%%%%%%%%%%%%%%%%%%%
\section{Partial $R^2$ Parameterization}
\label{sec:r2}

We now derive a scale-free reparameterization using partial $R^2$ values, following the approach of \citet{cinelli2020making}.

\subsection{Residualized Variables}

Let $\Pi_X[\cdot]$ denote the $L_2$ projection onto functions of $X$.
Define the residualized variables:
\begin{align}
    \widetilde{U} &:= U - \Pi_X[U] = U - \E[U \mid X] \\
    \widetilde{S} &:= S - \Pi_X[S] = S - \Pp(S = 1 \mid X) \\
    \widetilde{\tau} &:= \tau - \Pi_X[\tau]
\end{align}
where $\tau := Y(1) - Y(0)$ is the (latent) individual treatment effect and $\tau(X, U) := \E[\tau \mid X, U]$ is the full-information CATE.

Define residual variances:
\begin{align}
    \sigma_{\tau \mid X}^2 &:= \Var(\widetilde{\tau}) \\
    \sigma_{S \mid X}^2 &:= \Var(\widetilde{S}) = \Var(S)(1 - R^2_{S \sim X})
\end{align}
where $R^2_{S \sim X}$ is the population $R^2$ of $S$ regressed on $X$.
Since $S$ is binary, $\Var(S) = \pi(1-\pi)$ where $\pi := \Pp(S=1)$ under the reference distribution for $(S,X)$.

\subsection{Partial $R^2$ Sensitivity Parameters}

\begin{assumption}[Linear projection structure]
\label{assump:proj}
After residualizing on $X$, both the treatment effect and selection admit linear projections on $U$:
\begin{align}
    \widetilde{\tau} &= b \cdot \widetilde{U} + \varepsilon_\tau \\
    \widetilde{S} &= g \cdot \widetilde{U} + \varepsilon_S
\end{align}
with $\Cov(\widetilde{U}, \varepsilon_\tau) = \Cov(\widetilde{U}, \varepsilon_S) = 0$.
\end{assumption}

Define the partial $R^2$ parameters:
\begin{align}
\label{eq:partial-r2}
    R^2_{\tau \sim U \mid X} &:= \frac{\Var(b \widetilde{U})}{\Var(\widetilde{\tau})} = \frac{b^2 \Var(\widetilde{U})}{\sigma_{\tau \mid X}^2} \in [0, 1] \\
    R^2_{S \sim U \mid X} &:= \frac{\Var(g \widetilde{U})}{\Var(\widetilde{S})} = \frac{g^2 \Var(\widetilde{U})}{\sigma_{S \mid X}^2} \in [0, 1]
\end{align}

These have clear interpretations:
\begin{itemize}
    \item $R^2_{\tau \sim U \mid X}$: Proportion of residual treatment-effect variance explained by $U$ (after $X$-adjustment).
    \item $R^2_{S \sim U \mid X}$: Proportion of residual selection variance explained by $U$ (after $X$-adjustment).
\end{itemize}

\subsection{Partial $R^2$ Bound}

\begin{assumption}[Constant $X$-adjusted imbalance]
\label{assump:const-imbalance}
The moderator imbalance does not vary with $X$: $\Delta_U(X) \equiv \Delta_U^*$ almost surely.
\end{assumption}

\begin{remark}[When is constant imbalance reasonable?]
\cref{assump:const-imbalance} holds when trial participation shifts the \emph{mean} of $U$ by a constant amount after adjusting for $X$.
If $\Delta_U(X)$ varies with $X$, then the mapping from the bias term $\Delta_U^*=\E[\Delta_U(X)\mid S=0]$ to the selection partial-$R^2$ is no longer determined by $(R^2_{S\sim U\mid X},R^2_{S\sim X})$ alone; in that case, the raw $(\Gamma,\Lambda)$ bound in \cref{cor:simple-interval} remains valid.
\end{remark}

\begin{theorem}[Partial-$R^2$ bound for external-validity bias]
\label{thm:r2-bound}
Under \cref{assump:internal,assump:bridge,assump:linear,assump:proj,assump:const-imbalance} with constant $\beta(X) \equiv \beta$:
\begin{equation}
\label{eq:r2-bound}
    |\tau^* - \tau_X| \leq \sigma_{\tau \mid X} \sqrt{\frac{R^2_{\tau \sim U \mid X} \cdot R^2_{S \sim U \mid X}}{\Var(S)\,(1 - R^2_{S \sim X})}}
\end{equation}
\end{theorem}

\noindent\emph{Proof in \cref{app:proof-r2-bound}.}

\paragraph{Interpretation.}
The bound~\eqref{eq:r2-bound} separates three ingredients:
\begin{enumerate}
    \item $\sigma_{\tau \mid X}$: Residual treatment-effect heterogeneity after $X$-adjustment (a scale factor).
    \item $R^2_{\tau \sim U \mid X}$: How much of this heterogeneity could $U$ explain?
    \item $R^2_{S \sim U \mid X} / \{\Var(S)(1 - R^2_{S \sim X})\}$: How strongly could $U$ drive selection?
\end{enumerate}
Large bias requires substantial residual effect heterogeneity \emph{and} substantial residual selection, both attributable to the same $U$.

%%%%%%%%%%%%%%%%%%%%%%%%%%%%%%%%
% 6. ROBUSTNESS VALUES AND BENCHMARKING
%%%%%%%%%%%%%%%%%%%%%%%%%%%%%%%%
\section{Robustness Values and Benchmarking}
\label{sec:robustness}

\subsection{Robustness Values}

The partial-$R^2$ bound yields \emph{robustness values}---minimum confounding strength needed to change conclusions.

\begin{proposition}[Robustness value]
\label{prop:rv}
Let $B > 0$ be a target bias magnitude (e.g., $B = |\widehat{\tau}_X|$ to flip the sign).
Under \cref{thm:r2-bound}, any unobserved moderator must satisfy:
\begin{equation}
\label{eq:rv}
    R^2_{\tau \sim U \mid X} \cdot R^2_{S \sim U \mid X} \geq \Var(S)\,(1 - R^2_{S \sim X}) \left(\frac{B}{\sigma_{\tau \mid X}}\right)^2
\end{equation}
to induce bias at least $B$.
\end{proposition}

The \textbf{robustness value} (RV) is the right-hand side of~\eqref{eq:rv}.
A larger RV indicates greater robustness: conclusions can only change if an unobserved moderator simultaneously explains a large share of both residual treatment-effect variation and residual selection.

\subsection{Benchmarking Against Observed Covariates}

A powerful feature of the partial-$R^2$ parameterization is the ability to \emph{benchmark} against observed variables.

\paragraph{Procedure.}
Treat an observed covariate $Z \in X$ as if it were unobserved:
\begin{enumerate}
    \item Compute the partial $R^2$ of selection explained by $Z$ given $X_{-Z}$:
    \begin{equation}
        R^2_{S \sim Z \mid X_{-Z}} = \frac{R^2_{S \sim X} - R^2_{S \sim X_{-Z}}}{1 - R^2_{S \sim X_{-Z}}}
    \end{equation}

    \item Estimate the partial $R^2$ of treatment effect explained by $Z$ given $X_{-Z}$ (using effect modification regressions in the trial).

    \item Compare $R^2_{S \sim Z \mid X_{-Z}} \cdot R^2_{\tau \sim Z \mid X_{-Z}}$ to the robustness value RV.
\end{enumerate}

If the product for observed covariates is smaller than RV, then an unobserved moderator would need to be \emph{stronger} than any observed variable to overturn conclusions.

\begin{example}[Benchmark interpretation]
Suppose:
\begin{itemize}
    \item RV for sign reversal = 0.04
    \item Strongest observed moderator: age, with $R^2_{S \sim \text{age}} \cdot R^2_{\tau \sim \text{age}} = 0.02$
\end{itemize}
To reverse the sign, an unobserved moderator would need to be \emph{twice as strong} as age in its combined selection and effect-modification relationships.
\end{example}

%%%%%%%%%%%%%%%%%%%%%%%%%%%%%%%%
% 7. ESTIMATION AND PRACTICAL WORKFLOW
%%%%%%%%%%%%%%%%%%%%%%%%%%%%%%%%
\section{Estimation and Practical Workflow}
\label{sec:estimation}

\subsection{Complete Workflow}

\begin{algorithm}[tb]
\caption{OVB Sensitivity Analysis for Trial Generalization}
\label{alg:ovb}
\begin{algorithmic}[1]
\REQUIRE Trial data $\{(X_i, A_i, Y_i)\}_{i: S_i=1}$, target covariates $\{X_j\}_{j: S_j=0}$
\REQUIRE Sensitivity parameters: $(\Gamma, \Lambda)$ or $(R^2_{\tau \sim U \mid X}, R^2_{S \sim U \mid X})$
\STATE Fit outcome model on trial: $\widehat{\mu}_a(x)$
\STATE Compute baseline estimate: $\widehat{\tau}_X = \frac{1}{n_o} \sum_{j: S_j=0} [\widehat{\mu}_1(X_j) - \widehat{\mu}_0(X_j)]$
\STATE Construct baseline CI: $\text{CI}_{\text{base}}$ via bootstrap
\STATE Estimate $R^2_{S \sim X}$ from pooled $(S, X)$ data
\STATE Estimate $\widehat{\sigma}_{\tau \mid X}$ (see \cref{sec:sigma-tau})
\STATE Compute bias bound $b$ using \cref{cor:simple-interval} or \cref{eq:r2-bound}
\STATE \textbf{Output:} Sensitivity interval $\text{CI}_{\text{sens}} = [\underline{\tau} - b, \overline{\tau} + b]$
\end{algorithmic}
\end{algorithm}

\subsection{Estimating Component Quantities}

\paragraph{Baseline transported estimate.}
Any standard generalization estimator can be used: outcome modeling, IPW, or doubly robust.
We recommend cross-fitting for valid inference \citep{dahabreh2019extending}.

\paragraph{$R^2_{S \sim X}$.}
Estimable from pooled $(S, X)$ data via linear probability model or logistic regression (using pseudo-$R^2$).

\paragraph{$\sigma_{\tau \mid X}$.}
\label{sec:sigma-tau}
This residual CATE standard deviation is not point-identified without assumptions about individual treatment effects.
Options include:
\begin{enumerate}
    \item \textbf{Treat as sensitivity parameter}: Report curves over a plausible range.
    \item \textbf{Upper bound}: Use $\sqrt{\Var(Y \mid A=1, X) + \Var(Y \mid A=0, X)}$ from trial data.
    \item \textbf{Structural assumption}: Assume constant within-person correlation between potential outcomes.
\end{enumerate}

\subsection{Combining Sampling and Sensitivity Uncertainty}

Let $\text{CI}_{\text{base}} = [\underline{\tau}, \overline{\tau}]$ be a $(1-\alpha)$ confidence interval for $\tau_X$ under selection ignorability.
A conservative sensitivity interval is:
\begin{equation}
    \text{CI}_{\text{sens}} = [\underline{\tau} - b, \overline{\tau} + b]
\end{equation}
This ``inflate then report'' approach is standard in sensitivity analysis.
Sharper inference combining sampling and sensitivity uncertainty is possible but beyond our scope.

%%%%%%%%%%%%%%%%%%%%%%%%%%%%%%%%
% 8. EXPERIMENTS
%%%%%%%%%%%%%%%%%%%%%%%%%%%%%%%%
\section{Experiments}
\label{sec:experiments}

We evaluate the proposed sensitivity bounds in synthetic and semi-synthetic settings; full experimental details and additional results are reported in \cref{app:experiments}.

\paragraph{Controlled shift with known ground truth.}
We simulate a two-population design in which an unobserved effect modifier $U$ is distributed differently between the trial and target, inducing external-validity bias with oracle quantities available for validation (\cref{app:experiments}).

\paragraph{Coverage and calibration.}
Figure~\ref{fig:coverage} shows that the bias envelope is tight: coverage is 0\% when the moderation bound $\Gamma$ understates the true moderation strength and becomes valid once $\Gamma \ge \Gamma^*$.

\begin{figure}[t]
    \centering
    \includegraphics[width=0.95\columnwidth]{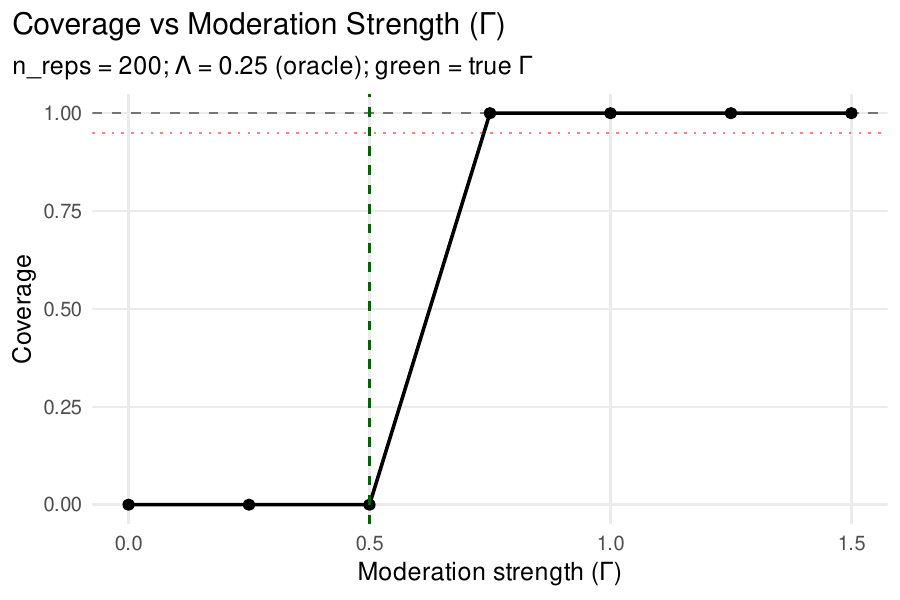}
    \caption{Monte Carlo coverage vs.\ moderation bound $\Gamma$ in a linear-Gaussian DGP. Coverage is 0\% for $\Gamma < \Gamma^*$ and jumps to 100\% at $\Gamma=\Gamma^*$ (green dashed).}
    \label{fig:coverage}
\end{figure}

\paragraph{Full confidence intervals.}
Figure~\ref{fig:fullCI} compares the sensitivity bias envelope alone versus combined with sampling uncertainty.
The ``Full CI'' (adding 95\% bootstrap intervals to the sensitivity bounds) achieves 100\% coverage even at $\Gamma=0$, since sampling uncertainty alone spans the true effect when the point estimate is unbiased on average.
This highlights that our sensitivity bounds quantify \emph{external-validity bias}, complementing rather than replacing standard inferential uncertainty.

\begin{figure}[t]
    \centering
    \includegraphics[width=0.95\columnwidth]{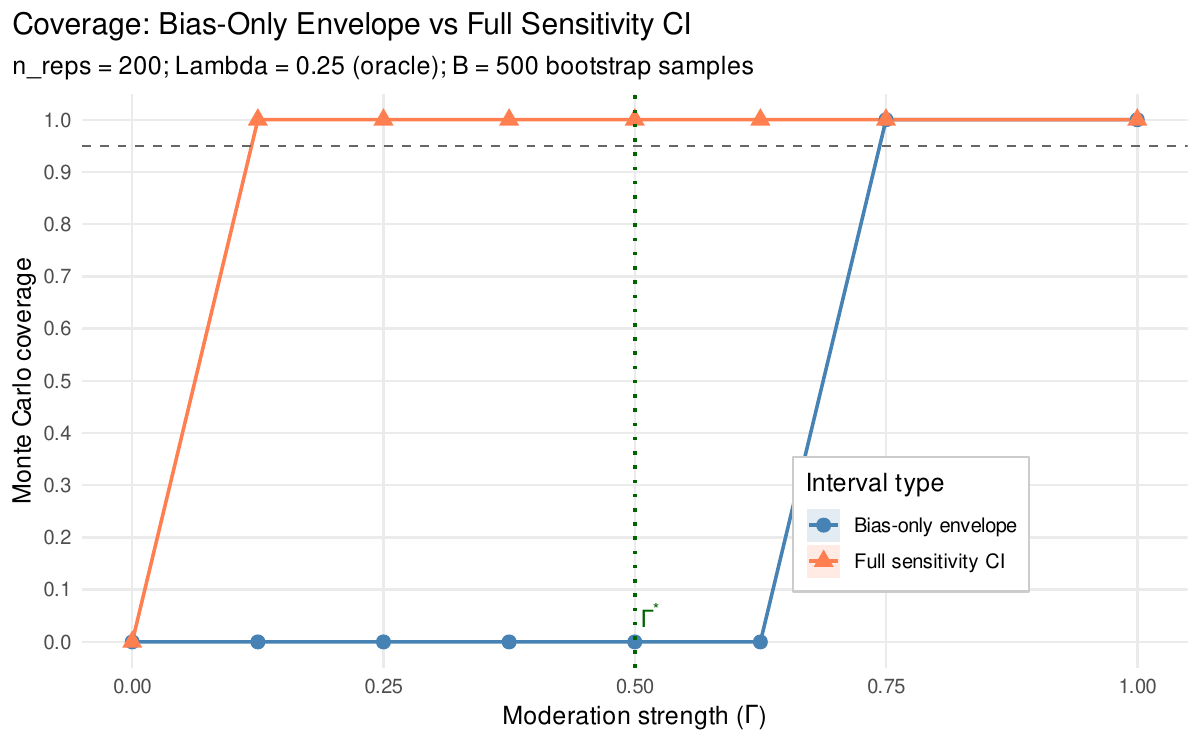}
    \caption{Coverage comparison: bias envelope only (solid) vs.\ full confidence interval combining sensitivity bounds with bootstrap uncertainty (dashed). The full CI achieves valid coverage even at $\Gamma=0$ due to sampling variability.}
    \label{fig:fullCI}
\end{figure}

\paragraph{Benchmarking and robustness.}
In a ``hide one moderator'' benchmark, observed-covariate partial-$R^2$ values provide a concrete scale for sensitivity parameters. Figure~\ref{fig:bench_scatter} compares these benchmarks to the sign-reversal robustness threshold; additional benchmarking plots appear in \cref{app:experiments}.

\begin{figure}[t]
    \centering
    \includegraphics[width=0.95\columnwidth]{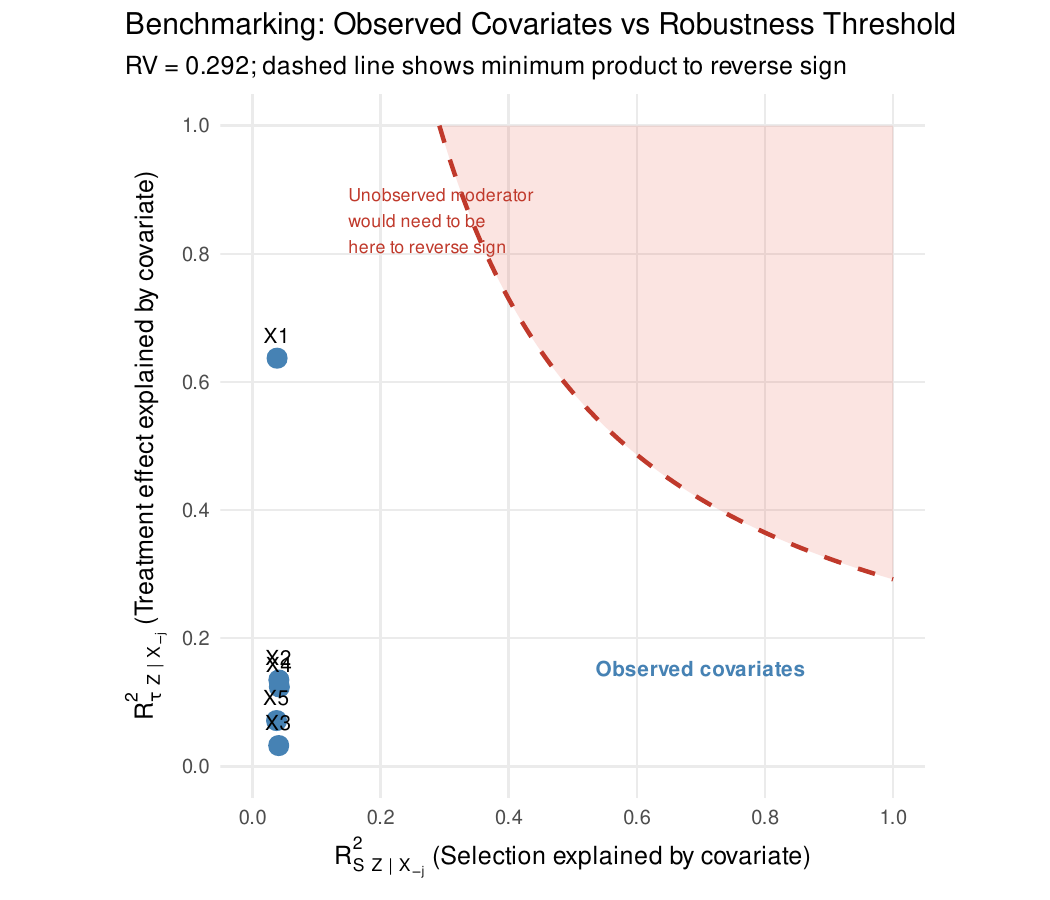}
    \caption{Benchmarking scatter plot. Each observed covariate is plotted by its partial $R^2$ for selection ($x$-axis) and treatment effect ($y$-axis); the dashed line is the sign-reversal robustness threshold.}
    \label{fig:bench_scatter}
\end{figure}

%%%%%%%%%%%%%%%%%%%%%%%%%%%%%%%%
% 9. DISCUSSION
%%%%%%%%%%%%%%%%%%%%%%%%%%%%%%%%
\section{Discussion}
\label{sec:discussion}

We develop a sensitivity analysis framework for trial generalization that decomposes external-validity bias into the \emph{strength} of effect modification by an unobserved variable and its \emph{imbalance} between trial and target populations.
This decomposition leads to closed-form sensitivity bounds and an interpretable calibration via robustness values and benchmarking against observed covariates.
Algorithm~\ref{alg:ovb} summarizes a practical workflow for reporting both a baseline transported estimate and a sensitivity envelope.

\paragraph{Limitations and modeling assumptions.}
Our tightest bounds rely on a linear effect-moderation bridge (\cref{assump:linear}) and, for the partial-$R^2$ mapping, a constant $X$-adjusted imbalance (\cref{assump:const-imbalance}).
The linear model assumption is restrictive but serves as a transparent baseline; when treatment effect heterogeneity is approximately linear in $U$ conditional on $X$, the bounds remain informative.
The constant imbalance assumption simplifies the relationship between the two sensitivity parameters, enabling a single-parameter robustness summary.

In practice, researchers should view our bounds as providing a structured sensitivity analysis rather than exact confidence intervals.
When the assumptions are violated, the bounds may be either conservative or anti-conservative depending on the nature of the misspecification.
Monte Carlo studies (Section~\ref{sec:experiments}) suggest the bounds perform well even under moderate departures from linearity.

\paragraph{Practical guidance.}
We recommend the following workflow for applied researchers:
(i) report the baseline transported estimate $\hat\tau_X$ with confidence intervals accounting for sampling uncertainty;
(ii) compute robustness values $\text{RV}^0$ (sign reversal) and $\text{RV}^q$ (clinical threshold) to summarize how strong an unobserved moderator must be to change conclusions;
(iii) benchmark these values against observed covariates to assess plausibility;
(iv) present sensitivity contour plots showing how the interval expands across a grid of $(R^2_{\tau \sim U \mid X}, R^2_{S \sim U \mid X})$ values.
This approach separates the empirical estimation (step i) from the sensitivity analysis (steps ii--iv), making the role of assumptions transparent.

\paragraph{Connections to general OVB theory.}
A complementary route to sensitivity analysis views trial generalization as a covariate-shift problem and applies the general OVB framework of \citet{chernozhukov2024ovb}, which expresses bias as the covariance between regression and Riesz representer approximation errors.
\cref{app:riesz} sketches this connection; developing sharp RR-based bounds with modern ML nuisance estimators is a natural extension.

\paragraph{Future directions.}
Several extensions merit investigation:
(i) allowing $X$-varying moderation coefficients $\beta(X)$ while maintaining tractable bounds;
(ii) incorporating sampling uncertainty in the sensitivity parameters themselves;
(iii) extending to multi-arm trials and factorial designs;
(iv) developing sensitivity analyses for subgroup-specific treatment effects.
The OVB decomposition provides a foundation for these extensions by clearly separating the sources of external-validity bias.

%%%%%%%%%%%%%%%%%%%%%%%%%%%%%%%%
% ACKNOWLEDGEMENTS (hidden in review)
%%%%%%%%%%%%%%%%%%%%%%%%%%%%%%%%
% \section*{Acknowledgements}
% Omitted for anonymous submission.

%%%%%%%%%%%%%%%%%%%%%%%%%%%%%%%%
% IMPACT STATEMENT
%%%%%%%%%%%%%%%%%%%%%%%%%%%%%%%%
\section*{Impact Statement}

This paper presents work whose goal is to advance methods for assessing the external validity of causal effect estimates from randomized trials.
By providing tools to quantify uncertainty about generalization, we aim to improve evidence-based decision-making in medicine, policy, and technology deployment.
There are many potential societal consequences of our work, none of which we feel must be specifically highlighted here.

%%%%%%%%%%%%%%%%%%%%%%%%%%%%%%%%
% REFERENCES
%%%%%%%%%%%%%%%%%%%%%%%%%%%%%%%%
\bibliography{icml_ovb_generalization}
\bibliographystyle{icml2026}

%%%%%%%%%%%%%%%%%%%%%%%%%%%%%%%%
% APPENDIX
%%%%%%%%%%%%%%%%%%%%%%%%%%%%%%%%
\newpage
\appendix
\onecolumn

\section{Detailed Proofs}
\label{app:proofs}

\subsection{Proof of \cref{lem:ovb-identity} (Detailed)}
\label{app:proof-ovb-identity}

We provide a step-by-step derivation of the OVB identity.

\begin{proof}
\textbf{Step 1: Conditional expectations under the linear model.}

Fix an arbitrary covariate value $x$ and population indicator $s \in \{0, 1\}$.
By \cref{assump:bridge}, $Y(a) \indep S \mid (X, U)$, so the conditional distribution of $Y(a)$ given $(X, U)$ is the same in both populations.

By the law of iterated expectations:
\begin{equation}
    \E[Y(a) \mid X = x, S = s] = \E[\E[Y(a) \mid X = x, U, S = s] \mid X = x, S = s]
\end{equation}

By \cref{assump:linear}:
\begin{equation}
    \E[Y(a) \mid X = x, U, S] = m_a(x) + \eta_a(x) \cdot U
\end{equation}

Substituting:
\begin{align}
    \E[Y(a) \mid X = x, S = s] &= \E[m_a(x) + \eta_a(x) \cdot U \mid X = x, S = s] \\
    &= m_a(x) + \eta_a(x) \cdot \E[U \mid X = x, S = s]
\end{align}

\textbf{Step 2: CATE in each population.}

The conditional average treatment effect at $X = x$ in population $s$ is:
\begin{align}
    \tau_s(x) &:= \E[Y(1) - Y(0) \mid X = x, S = s] \\
    &= \E[Y(1) \mid X = x, S = s] - \E[Y(0) \mid X = x, S = s] \\
    &= [m_1(x) + \eta_1(x) \cdot \E[U \mid X = x, S = s]] \\
    &\quad - [m_0(x) + \eta_0(x) \cdot \E[U \mid X = x, S = s]] \\
    &= [m_1(x) - m_0(x)] + [\eta_1(x) - \eta_0(x)] \cdot \E[U \mid X = x, S = s] \\
    &= \tau_0(x) + \beta(x) \cdot \E[U \mid X = x, S = s]
\end{align}
where $\tau_0(x) := m_1(x) - m_0(x)$ is the baseline CATE (the CATE if $U = 0$) and $\beta(x) := \eta_1(x) - \eta_0(x)$ is the moderation strength.

\textbf{Step 3: Difference between populations.}

The CATE in the target population is:
\begin{equation}
    \tau_{S=0}(x) := \E[Y(1) - Y(0) \mid X = x, S = 0] = \tau_0(x) + \beta(x) \cdot \E[U \mid X = x, S = 0]
\end{equation}

The CATE in the trial population is:
\begin{equation}
    \tau_{S=1}(x) := \E[Y(1) - Y(0) \mid X = x, S = 1] = \tau_0(x) + \beta(x) \cdot \E[U \mid X = x, S = 1]
\end{equation}

The difference is:
\begin{align}
    \tau_{S=0}(x) - \tau_{S=1}(x) &= \beta(x) \cdot [\E[U \mid X = x, S = 0] - \E[U \mid X = x, S = 1]] \\
    &= \beta(x) \cdot \Delta_U(x)
\end{align}
where $\Delta_U(x) := \E[U \mid X = x, S = 0] - \E[U \mid X = x, S = 1]$ is the moderator imbalance at $X = x$.

\textbf{Step 4: Averaging over the target distribution.}

The TATE is:
\begin{equation}
    \tau^* = \E[\tau_{S=0}(X) \mid S = 0] = \E[\tau_0(X) + \beta(X) \cdot \E[U \mid X, S = 0] \mid S = 0]
\end{equation}

The $X$-adjusted transport estimand is:
\begin{equation}
    \tau_X = \E[\tau_{S=1}(X) \mid S = 0] = \E[\tau_0(X) + \beta(X) \cdot \E[U \mid X, S = 1] \mid S = 0]
\end{equation}

Note that $\tau^r(X) = \tau_1(X)$ is the trial CATE at $X$.

Subtracting:
\begin{align}
    \tau^* - \tau_X &= \E[\tau_0(X) - \tau_1(X) \mid S = 0] \\
    &= \E[\beta(X) \cdot \Delta_U(X) \mid S = 0]
\end{align}

Rearranging gives \cref{eq:tate-ovb}:
\begin{equation}
    \tau^* = \tau_X + \E[\beta(X) \cdot \Delta_U(X) \mid S = 0]
\end{equation}

\textbf{Step 5: Constant moderation case.}

If $\beta(X) \equiv \beta$ is constant, factor it out:
\begin{align}
    \tau^* - \tau_X &= \E[\beta \cdot \Delta_U(X) \mid S = 0] \\
    &= \beta \cdot \E[\Delta_U(X) \mid S = 0] \\
    &= \beta \cdot \Delta_U^*
\end{align}
where $\Delta_U^* := \E[\Delta_U(X) \mid S = 0]$ is the average moderator imbalance.
This gives \cref{eq:tate-ovb-const}.
\end{proof}

\subsection{Proof of \cref{thm:r2-bound} (Detailed)}
\label{app:proof-r2-bound}

\begin{proof}
\textbf{Step 1: Normalize $U$.}

Without loss of generality, assume $U$ is scaled so that $\Var(\widetilde{U}) = \Var(U - \E[U \mid X]) = 1$.
This is a normalization that simplifies the algebra.

\textbf{Step 2: Relate $R^2_{\tau \sim U \mid X}$ to $b$.}

Under \cref{assump:proj}, after residualizing on $X$:
\begin{equation}
    \widetilde{\tau} = b \cdot \widetilde{U} + \varepsilon_\tau
\end{equation}
where $\Cov(\widetilde{U}, \varepsilon_\tau) = 0$.

Taking variances:
\begin{equation}
    \Var(\widetilde{\tau}) = b^2 \Var(\widetilde{U}) + \Var(\varepsilon_\tau) = b^2 + \Var(\varepsilon_\tau)
\end{equation}

The partial $R^2$ is:
\begin{equation}
    R^2_{\tau \sim U \mid X} = \frac{\Var(b \widetilde{U})}{\Var(\widetilde{\tau})} = \frac{b^2}{\sigma_{\tau \mid X}^2}
\end{equation}

Solving for $|b|$:
\begin{equation}
    |b| = \sigma_{\tau \mid X} \sqrt{R^2_{\tau \sim U \mid X}}
\end{equation}

\textbf{Step 3: Relate $R^2_{S \sim U \mid X}$ to $g$.}

Similarly, under \cref{assump:proj}:
\begin{equation}
    \widetilde{S} = g \cdot \widetilde{U} + \varepsilon_S
\end{equation}

The residual variance of $S$ is:
\begin{equation}
    \Var(\widetilde{S}) = \Var(S - \E[S \mid X]) = \Var(S)(1 - R^2_{S \sim X})
\end{equation}

The partial $R^2$ is:
\begin{equation}
    R^2_{S \sim U \mid X} = \frac{\Var(g \widetilde{U})}{\Var(\widetilde{S})} = \frac{g^2}{\Var(S)(1 - R^2_{S \sim X})}
\end{equation}

Solving for $|g|$:
\begin{equation}
    |g| = \sqrt{R^2_{S \sim U \mid X}\,\Var(S)(1 - R^2_{S \sim X})}
\end{equation}

\textbf{Step 4: Compute the imbalance coefficient.}

The coefficient $\delta$ in the linear projection of $\widetilde{U}$ on $\widetilde{S}$ is:
\begin{equation}
    \delta = \frac{\Cov(\widetilde{U}, \widetilde{S})}{\Var(\widetilde{S})}
\end{equation}

Under \cref{assump:proj}:
\begin{equation}
    \Cov(\widetilde{U}, \widetilde{S}) = \Cov(\widetilde{U}, g \widetilde{U} + \varepsilon_S) = g \Var(\widetilde{U}) = g
\end{equation}

Therefore:
\begin{equation}
    \delta = \frac{g}{\Var(\widetilde{S})} = \frac{g}{\Var(S)(1 - R^2_{S \sim X})}
\end{equation}

Taking absolute values:
\begin{equation}
    |\delta| = \frac{|g|}{\Var(S)(1 - R^2_{S \sim X})} = \sqrt{\frac{R^2_{S \sim U \mid X}}{\Var(S)(1 - R^2_{S \sim X})}}
\end{equation}

Under \cref{assump:const-imbalance} and the centering $\E[U\mid X]=0$ from \cref{assump:linear}, we have $\Delta_U(X)\equiv\Delta_U^*=-\delta$.

\textbf{Step 5: Combine for the bias bound.}

From \cref{lem:ovb-identity} with constant $\beta$:
\begin{equation}
    |\tau^* - \tau_X| = |\beta| \cdot |\Delta_U^*|
\end{equation}
\noindent
Under \cref{assump:linear} with constant $\beta$, the linear projection coefficient satisfies $b=\beta$, and under \cref{assump:const-imbalance} we have $\Delta_U^*=-\delta$. Therefore:
\begin{align}
    |\tau^* - \tau_X| &= |b| \cdot |\delta| \\
    &= \sigma_{\tau \mid X} \sqrt{R^2_{\tau \sim U \mid X}} \cdot \sqrt{\frac{R^2_{S \sim U \mid X}}{\Var(S)(1 - R^2_{S \sim X})}} \\
    &= \sigma_{\tau \mid X} \sqrt{\frac{R^2_{\tau \sim U \mid X} \cdot R^2_{S \sim U \mid X}}{\Var(S)(1 - R^2_{S \sim X})}}
\end{align}

This is \cref{eq:r2-bound}.
\end{proof}

\section{Connection to General OVB Theory via Riesz Representers}
\label{app:riesz}

Our main text focuses on a \emph{linear} effect-moderation bridge, which yields the transparent identity
\(\tau^* - \tau_X = \E[\beta(X)\Delta_U(X)\mid S=0]\).
A complementary (and more general) route is to view trial generalization as a \emph{covariate-shift} problem and apply the general omitted-variable-bias (OVB) framework of \citet{chernozhukov2024ovb}.
We summarize the key mapping here.

\subsection{TATE as a linear functional}

Fix a treatment arm \(a\in\{0,1\}\) and define the long conditional mean in the trial
\[
m_a^{\ell}(x,u) := \E[Y \mid A=a, X=x, U=u, S=1],
\]
and the short regression that omits \(U\),
\[
m_a^{s}(x) := \E[Y \mid A=a, X=x, S=1] = \E[m_a^{\ell}(X,U)\mid X=x,S=1].
\]
Under \cref{assump:bridge}, the same conditional mean function describes potential outcomes in the target as well.

The target mean potential outcome can be written as the linear functional
\[
\theta_a^* := \E[Y(a)\mid S=0] = \E[m_a^{\ell}(X,U)\mid S=0],
\]
and the TATE is \(\tau^*=\theta_1^*-\theta_0^*\).

\subsection{Riesz representers and the role of trial participation}

Parameters like \(\theta_a^*\) can be written as inner products between a regression function and a \emph{Riesz representer} (RR).
Intuitively, the RR plays the role of a weighting function that ``moves'' expectations from the observed trial distribution to the target distribution.
In our setting, the RR is closely related to inverse odds of trial participation weights.
When we omit \(U\), both the regression \(m_a^{\ell}\) and the RR can change, because \(U\) can explain additional variation in outcomes and in selection into the trial.

\subsection{General OVB bound and partial \(R^2\) reparameterization}

\citet{chernozhukov2024ovb} show that for a broad class of causal targets (including covariate-shift policy effects), the OVB between a ``short'' functional (omitting \(U\)) and a ``long'' functional (including \(U\)) admits the generic form
\[
\text{OVB} = \Cov(\varepsilon_m,\varepsilon_{\alpha}),
\]
where \(\varepsilon_m\) and \(\varepsilon_{\alpha}\) are approximation errors in the regression function and in the RR, respectively.
By Cauchy--Schwarz,
\[
|\text{OVB}| \le \|\varepsilon_m\|_{L_2}\,\|\varepsilon_{\alpha}\|_{L_2},
\]
and these \(L_2\) norms can be reparameterized in terms of \emph{partial \(R^2\)} measures that quantify the incremental explanatory power of \(U\) for (i) the outcome regression and (ii) the RR.
This yields scale-free sensitivity bounds that closely mirror \cref{eq:r2-bound}, while allowing for nonlinear nuisance functions estimated by modern machine learning.

Developing the sharpest RR-based bounds and inference specifically tailored to \(\tau^*\) in non-nested trial generalization designs is a natural direction for our follow-up journal paper.

\section{Additional Experimental Details}
\label{app:experiments}

\subsection{Simulation 1: Controlled Linear-Gaussian DGP}

\paragraph{Data-generating process.}
We design a two-population simulation that directly induces an external-validity violation with known ground truth, satisfying \cref{assump:linear} exactly.
The trial population ($S=1$, denoted $r$) and target population ($S=0$, denoted $o$) have:
\begin{align}
    X^r &\sim \mathcal{N}(0, I_p), \quad X^o \sim \mathcal{N}(\mu_{\text{shift}} \mathbf{1}_p, I_p), \\
    U^s \mid X &\sim \mathcal{N}(\gamma_s X_1, 1), \quad s \in \{r, o\}.
\end{align}
The key feature is that the unobserved moderator $U$ has a different conditional distribution given $X$ in the two populations: in the trial, $\E[U \mid X, S=1] = \gamma_r X_1$, while in the target, $\E[U \mid X, S=0] = \gamma_o X_1$.
This creates the moderator imbalance $\Delta_U(X) = (\gamma_o - \gamma_r)X_1$ that drives external-validity bias.

Treatment is randomized within the trial: $A \sim \text{Bernoulli}(0.5)$ given $S=1$.
Potential outcomes follow a linear model with $U$ as an effect modifier:
\begin{equation}
    Y(a) = \beta_0 + \beta_X^\top X + a \cdot (\tau_0 + \beta_U \cdot U) + \varepsilon, \quad \varepsilon \sim \mathcal{N}(0, \sigma^2).
\end{equation}

\paragraph{Parameter settings.}
We set $n_{\text{trial}} = 2000$, $n_{\text{target}} = 5000$, $p = 5$, $\mu_{\text{shift}} = 0.5$, $\gamma_r = 0$, $\gamma_o = 0.5$, $\beta_U = 0.5$, $\tau_0 = 1$, and $\sigma=1$.
This yields oracle moderation strength $\Gamma^*=\beta_U=0.5$ and oracle average imbalance $\Delta_U^* = (\gamma_o-\gamma_r)\mu_{\text{shift}}=0.25$, with true bias $0.125$.
In this Gaussian DGP, $\Delta_U(X)$ is unbounded, so the almost-sure bound $\Lambda$ in \cref{cor:simple-interval} does not exist; we report results using $\Delta_U^*$ as the imbalance scale.

\paragraph{Sensitivity envelope.}
Figure~\ref{fig:envelope} illustrates the OVB sensitivity envelope for a single realization.

\begin{figure}[t]
    \centering
    \includegraphics[width=0.8\textwidth]{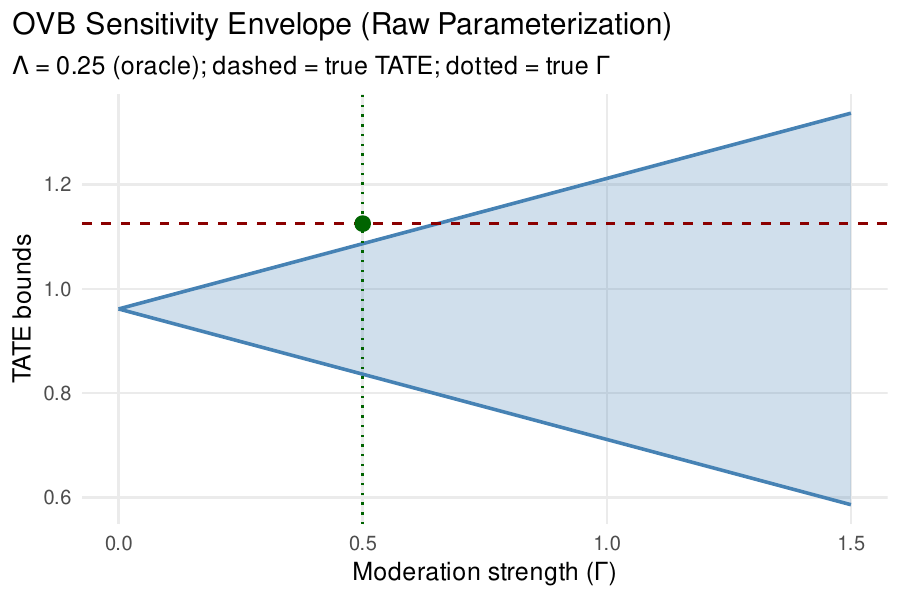}
    \caption{OVB sensitivity envelope for a single simulation. The baseline estimate $\hat{\tau}_X$ is biased for the true TATE $\tau^*$, and the sensitivity interval expands with $\Gamma$.}
    \label{fig:envelope}
\end{figure}

\paragraph{Bias-only envelope vs full sensitivity CI.}
The bias-only envelope accounts for systematic bias from omitted moderators but not sampling uncertainty in $\hat\tau_X$.
Combining the OVB bound with a confidence interval for $\hat\tau_X$ yields a full sensitivity CI; Figure~\ref{fig:fullci} compares the two approaches.

\begin{figure}[t]
    \centering
    \includegraphics[width=0.8\textwidth]{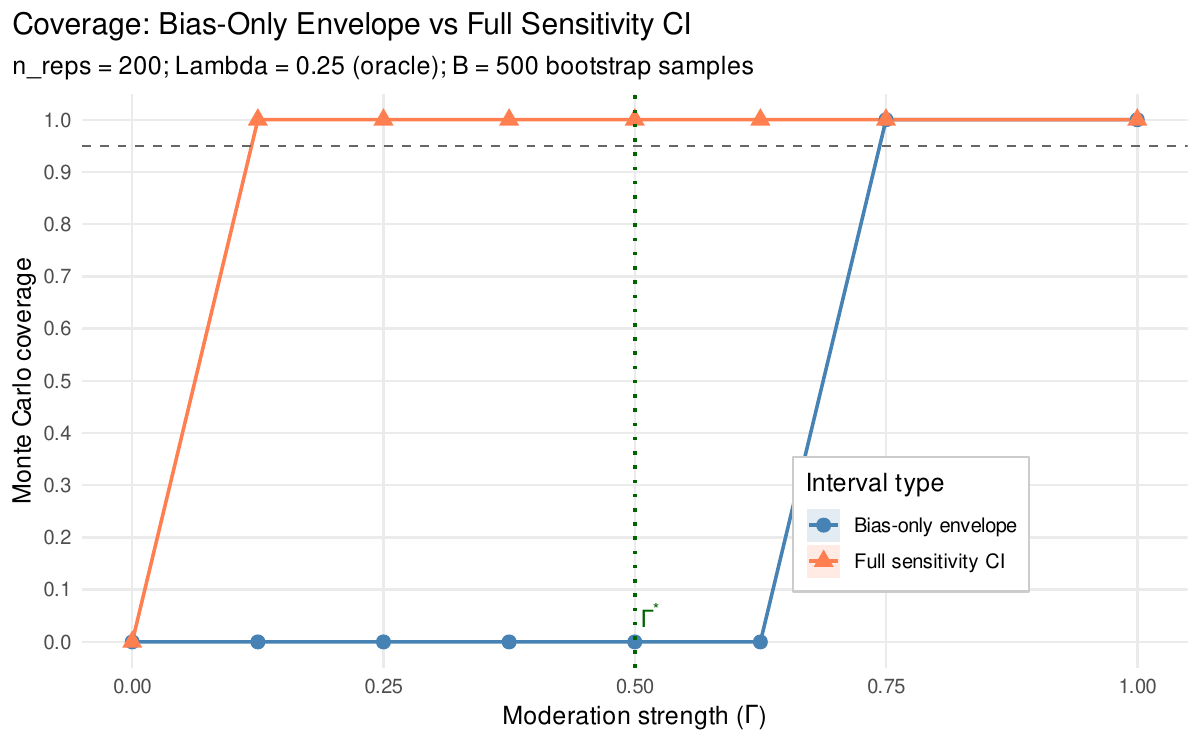}
    \caption{Comparison of bias-only envelope vs full sensitivity CI coverage. The full CI incorporates bootstrap uncertainty in $\hat\tau_X$ in addition to the OVB bias bound.}
    \label{fig:fullci}
\end{figure}

\subsection{Simulation 2: Nonlinear DGP}

We relax the linear assumptions to test robustness using a nonlinear baseline and heterogeneous moderation:
\begin{itemize}
    \item Nonlinear baseline: $m(X) = \beta_0 + \sin(\pi X_1/2) + 0.5 X_2^2$
    \item Heterogeneous moderation: $\beta(X) = \beta_{\text{base}} + \beta_X X_1$
\end{itemize}
This violates the linear effect-modification model in \cref{assump:linear} and tests whether the bounds remain informative under misspecification.

\begin{figure}[t]
    \centering
    \includegraphics[width=0.8\textwidth]{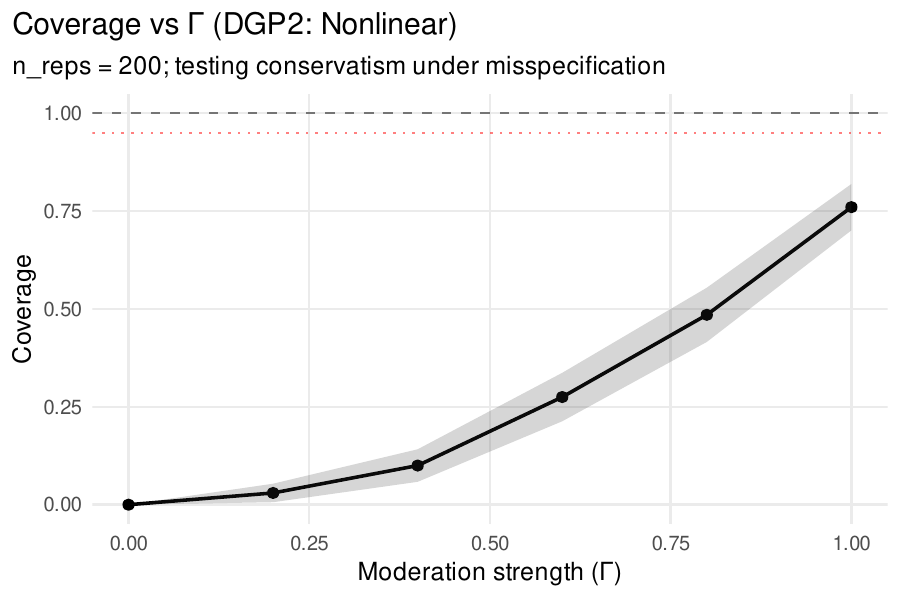}
    \caption{Coverage under a nonlinear DGP with heterogeneous moderation. Coverage increases with $\Gamma$ but does not reach 95\% at $\Gamma=1$, reflecting that a constant bound cannot fully capture $X$-varying moderation strength.}
    \label{fig:dgp2}
\end{figure}

\subsection{Simulation 3: High-Dimensional ML Setting}

We consider a high-dimensional simulation with $p=50$ covariates (only 10 relevant), nonlinear baseline outcomes, and logistic selection into the trial.
We compare linear regression with interactions, LASSO, and random forest as baseline estimators.

\begin{figure}[t]
    \centering
    \includegraphics[width=0.8\textwidth]{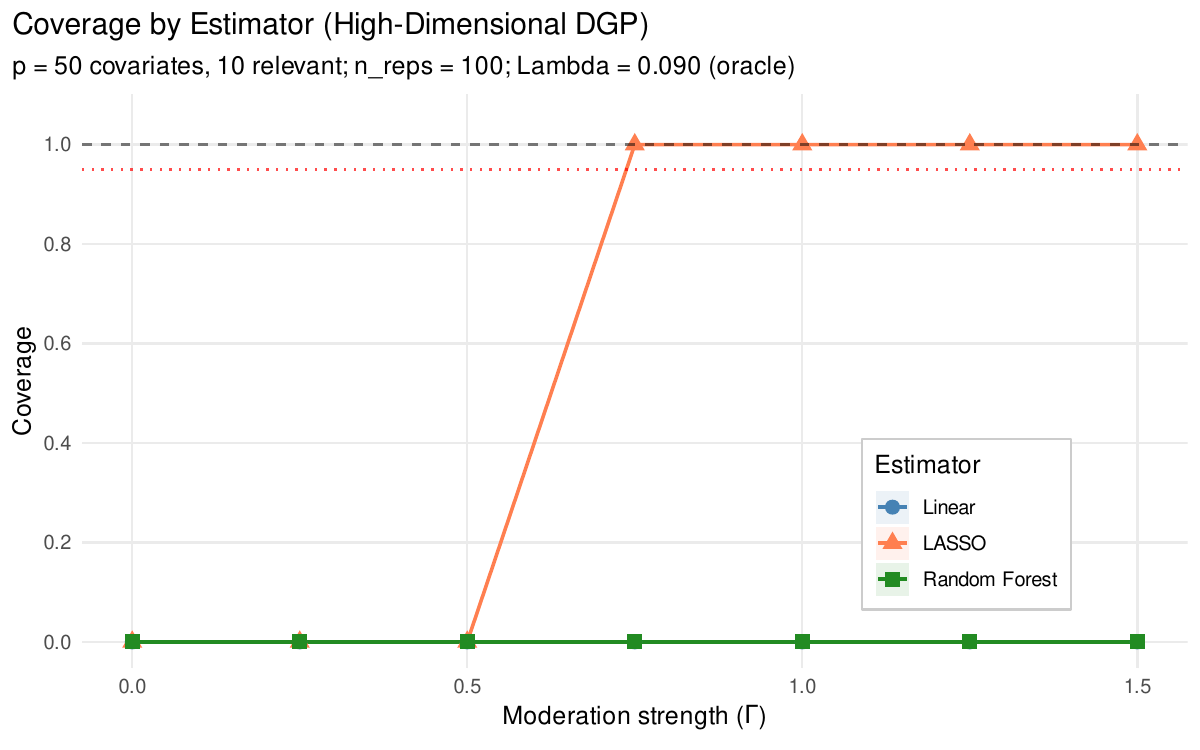}
    \caption{Coverage in a high-dimensional setting. OVB bounds correct for omitted moderators but do not correct baseline estimator misspecification.}
    \label{fig:ml}
\end{figure}

\subsection{Comparison with Marginal Sensitivity Model}

We compare our OVB envelope to a marginal sensitivity model (MSM) that bounds selection odds ratios rather than moderator strength.

\begin{figure}[t]
    \centering
    \includegraphics[width=0.8\textwidth]{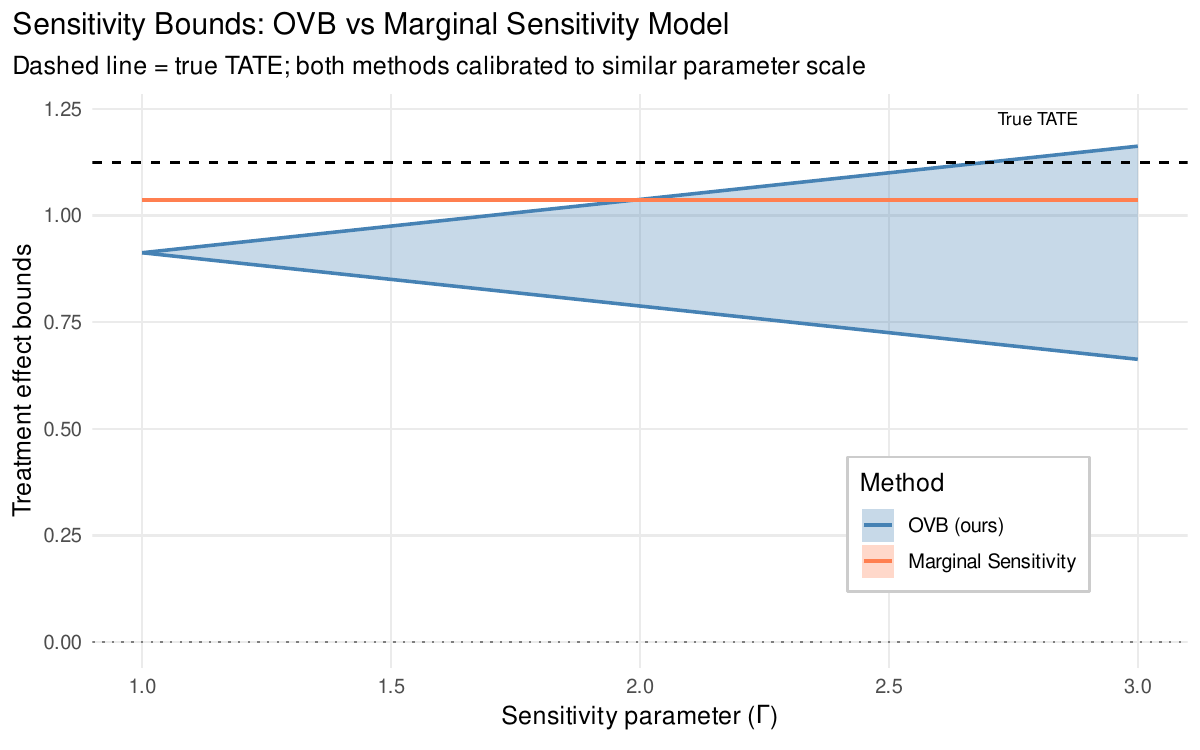}
    \caption{Comparison of OVB sensitivity bounds vs marginal sensitivity model bounds.}
    \label{fig:comparison}
\end{figure}

\subsection{Semi-Synthetic Benchmark and Robustness Visualization}

We demonstrate benchmarking via a ``hide one moderator'' experiment: we generate data with five observed effect modifiers and then treat one covariate as unobserved, using its estimated partial-$R^2$ values as sensitivity parameters.
Figure~\ref{fig:benchmark} shows the resulting benchmark scatter plot.
Figure~\ref{fig:rv_contour} visualizes the sign-reversal region for $(R^2_{\tau \sim U \mid X}, R^2_{S \sim U \mid X})$; Figure~\ref{fig:bench_scatter} in the main text overlays observed covariates against the robustness threshold.

\begin{figure}[t]
    \centering
    \includegraphics[width=0.8\textwidth]{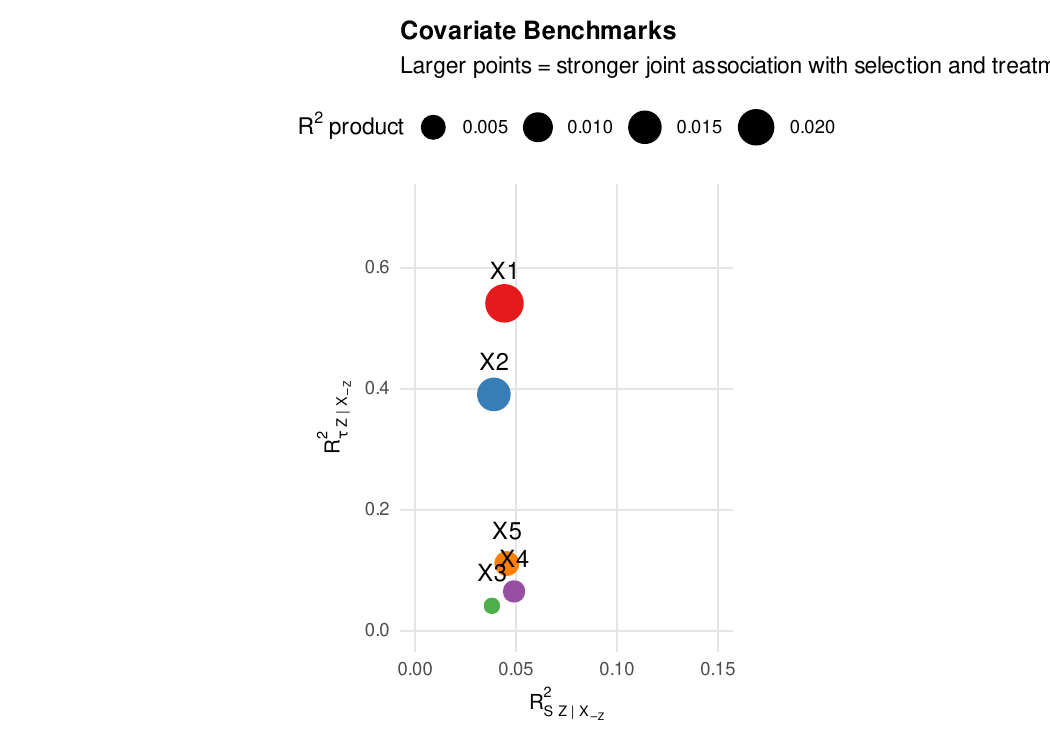}
    \caption{Covariate benchmarks for OVB sensitivity. Each point is an observed covariate, plotted by its partial $R^2$ with selection and with treatment effect.}
    \label{fig:benchmark}
\end{figure}

\begin{figure}[t]
    \centering
    \includegraphics[width=0.8\textwidth]{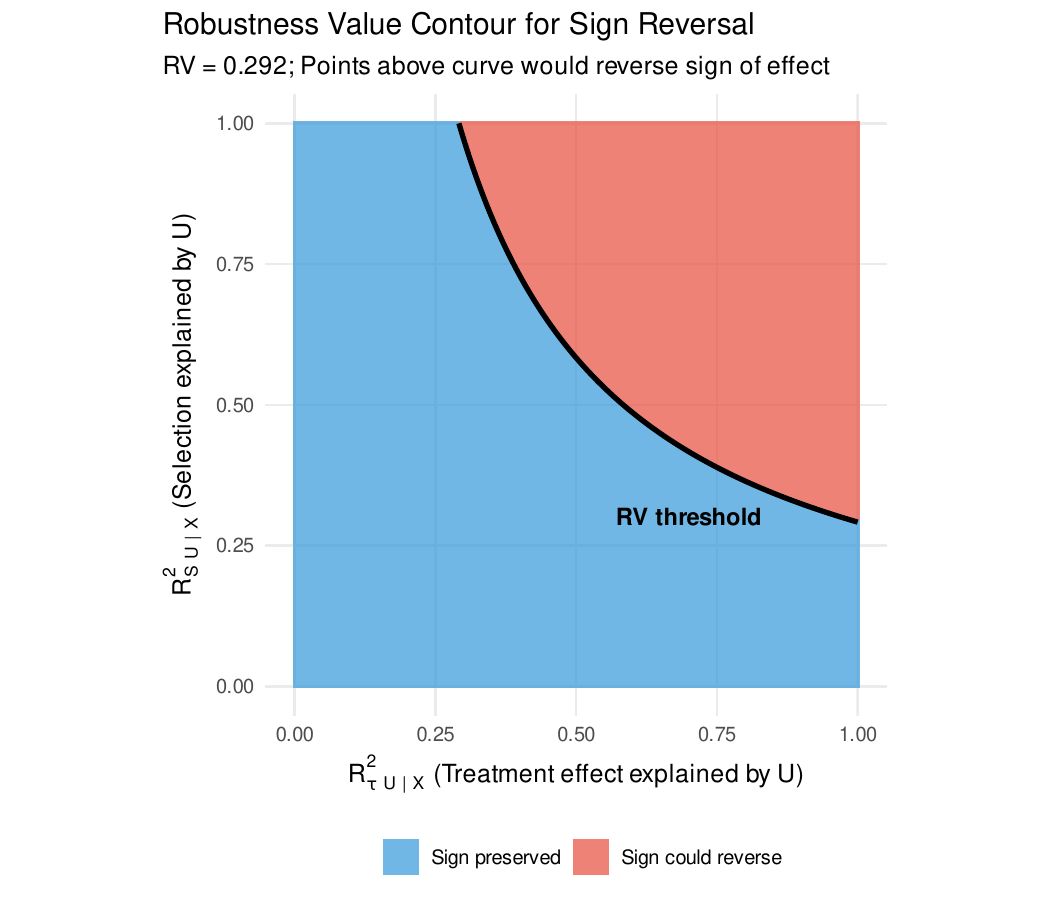}
    \caption{Robustness value contour. The red region shows $(R^2_{\tau \sim U \mid X}, R^2_{S \sim U \mid X})$ combinations that would induce sign reversal.}
    \label{fig:rv_contour}
\end{figure}

\subsection{Computational Details}

All experiments use 500 replications.
Parallel computation uses $\min(n_{\text{cores}} - 1, 8)$ cores.
Bootstrap confidence intervals use 1000 resamples.

\end{document}